\documentclass[11pt]{article}
\usepackage{graphicx}
\usepackage{hyperref}
\hypersetup{colorlinks,allcolors=black}
\usepackage{amsthm,amsmath}
\usepackage{enumerate}
\usepackage{amsfonts}
\usepackage{bbold}
\usepackage{xcolor}
\usepackage{mathtools}

\newtheorem{theorem}{Theorem}
\newtheorem{lemma}[theorem]{Lemma}

\newtheorem{corollary}[theorem]{Corollary}
\newtheorem{definition}[theorem]{Definition}

\theoremstyle{definition}
\newtheorem{rem}[theorem]{Remark}

\numberwithin{equation}{section}  
\numberwithin{theorem}{section}  

\DeclareMathOperator{\sgn}{sgn}

\DeclarePairedDelimiter{\norm}{||}{||_2}

\newcommand{\ep}{\varepsilon}

\newcommand{\cP}{\mathcal{P}}

\title{Strong quantum state transfer on graphs via loop edges}
\author{Gabor Lippner, Yujia Shi}

\begin{document}

\maketitle

\begin{abstract}
We quantify the effect of weighted loops at the source and target nodes of a graph on the strength of quantum state transfer between these vertices. We give lower bounds on loop weights that guarantee strong transfer fidelity that works for any graph where this protocol is feasible. By considering local spectral symmetry, we show that the required weight size depends only on the maximum degree of the graph and, in some less favorable cases, the distance between vertices. Additionally, we explore the duration for which transfer strength remains above a specified threshold.
\end{abstract}


\section{Introduction}

We study quantum walk on a simple connected graph. Our goal is to understand the effect of adding weighted loops on the source and target nodes on the transfer fidelity. Let $A$ denote the adjacency matrix of the network, and let $H = A + Q\cdot D_{u,v}$ where $D_{u,v}$ is the diagonal matrix with 1 at the $u$ and $v$ entries and 0 elsewhere. 

The continuous-time quantum walk is given by
\begin{equation}\label{eq:qw} U(t) = e^{itH},\end{equation} and, in particular, the $u\to v$ transfer strength at time $t$ is defined as 
\[ p_{u,v}(t) = p(t) = |U(t)_{u,v}|^2.\]
We are interested in $F(Q) = F(Q,u,v) = \sup_t p(t)$, the $u\to v$ \emph{transfer fidelity} of the network, as a function of $Q$.

The relevance of transfer strength as defined above comes from its connection to spin systems. Locally coupled networks of spin particles have been used extensively as a protocol to transfer quantum information over physical distances. See~\cite{Bose_2003} and~\cite{Christandl_2005}.
It is well-known~\cite{kay2010} that the time evolution (in the so-called 1-excitation subspace, which is most relevant for information transfer) of such a spin network is exactly described by the continuous-time quantum walk \eqref{eq:qw}.

Weighted loops on the graph, in this context, correspond to applying transverse magnetic fields to certain particles in the network~\cite{Shi_2005}.
Transfer strength in this protocol has also been numerically investigated in~\cite{casaccino2009quantum,Christandl_2005, xx_magnetic_2012, Chen_2016} 

In terms of rigorous results, it has been shown in~\cite{Tunneling_2012} that $\lim_{Q \to \infty} F(Q) = 1$ if and only if the network has a certain spectral symmetry around $u$ and $v$ (for more details see Section~\ref{se:decomposition}), but there were no quantitative bounds were given. Then, in~\cite{path} Kirkland and van Bommel gave very precise behavior of $F(Q)$ in the case of path graphs, $u,v$ being the endpoints of the path. It turned out that in order to achieve $F(Q) > 1-\ep$, one only needs $Q > \ep^{-1/2}$. What was striking, in particular, is that the threshold did not depend on the length of the path. Their method relies on a somewhat explicit computation of the eigenvalues and eigenvectors of the Hamiltonian $H$, and thus it doesn't immediately lend itself to analyze more complicated networks.

In~\cite{quantifying}, we devised a method to generalize the bounds for the path to arbitrary graphs that have an involution that maps $u$ to $v$. We showed that in order to achieve $F(Q) > 1-\ep$ on such graphs, it is sufficient to choose $Q > c(m,\ep)$ where $m$ is the maximum degree of the network. This method relied on the introduction of approximate eigenvectors instead of explicit formulas. While it resulted in a worse $\ep^{-2}$ - dependence on $\ep$, it still had the important feature of not depending on the size of the network, only its maximum degree.

In the current paper, we present a generalization of the quantitative bound to all graphs where $\lim_{Q\to \infty} F(Q) = 1$. For other graphs, the question in this form is meaningless as the fidelity converges to a number strictly less than 1. These new bounds depend on the maximum degree, and in certain cases also on the distance of $u$ and $v$. However, the improved methods yield a dependence of $1/\sqrt{\ep}$ in almost all cases. In particular, we recover the asymptotics of the Kirkland - van Bommel paper whenever $u,v$ are cospectral in the classical sense, including any graph with an involution.

\begin{theorem}\label{thm:main} Let $G(V,E)$ be a finite graph with maximum degree $m$, and $u,v \in V$ fixed. Assume $c \geq d$ where $c = co(u,v)$ denotes their \emph{cospectrality} (see Definition~\ref{def:cos}) and $d = d(u,v)$ denotes the distance of the two nodes. For any $\ep >0$ we have 
\[ Q > 16 \frac{1}{\ep^{1/\min(2,c-d+1)}} m^{1+\max(\frac{1}{2}, \frac{d}{c-d+1})}  \Longrightarrow F(Q) > 1- \ep. \]
\end{theorem}

\begin{rem}
In the above theorem, the dependence on the prescribed error is almost always $1/\sqrt{\ep}$, except for the $c=d$ case when it is $1/\ep$. The exponent of $m$ can be as large as $d+1$ when $c=d$ but goes down gradually as $c$ increases relative to $d$, reaching the optimal $3/2$ when $c \geq 3d-1$. 
In particular, when $u,v$ are cospectral in the usual sense (i.e., $co(u,v) = \infty$), then we get $F(Q) > 1-\ep$ as long as $Q = \Omega(\ep^{-1/2} m^{3/2})$.
\end{rem}

We also generalize our previous bounds~\cite{quantifying} on the readout time $t_0$ at which large transfer strength is guaranteed. While $t_0$ is exponential in the distance of the source and target nodes (as was the case in~\cite{path, quantifying} as well), we prove that within a large interval around $t_0$ the transfer strength remains close to 1. 

\begin{theorem}\label{thm:readout}
If $Q$ satisfies the lower bound of Theorem~\ref{thm:main} then there is a time $t_0 < 2\pi (Q+m)^{d-1}$ such that $p_{u,v}(t_0) \geq 1-\ep$. Furthermore, for any $\delta > 0$ if $t_0 (1-\delta/\pi) \leq t \leq t_0 (1+\delta/\pi)$ we have $p_{u,v}(t) \geq 1- \ep-2\delta$.  
\end{theorem}

The outline of this paper is as follows. In Section~\ref{se:decomposition} we outline the main strategy of the proof. In Section~\ref{sec:est}, we establish the key estimates on the components of the eigenvectors. Section~\ref{se:mainthm} concludes the proof of the main result using the estimates from previous sections. Finally, we address the bounds on the readout time and readout interval in Section~\ref{sec:readout}.

\subsection{Proof strategy}\label{se:decomposition}

The Hamiltonian $H$ is a real symmetric matrix so it has real eigenvalues $\lambda_1\geq\lambda_2\geq \dots \geq \lambda_n$. Let us denote the corresponding orthonormal eigenvectors by $\varphi_1, \varphi_2, \dots, \varphi_n$. Then we can compute the probability by 
\begin{align}\label{eq:p(t)}
    p(t) = |U(t)_{u,v}|^2 = \left|\sum_{j=1}^n \varphi_j(u)\varphi_j(v)e^{it\lambda_j}\right|^2.
\end{align}
If $Q$ is large enough, the Hamiltonian $H$ has two diagonal entries that are significantly larger than the rest. This results in the two largest eigenvalues $\lambda_1$ and $ \lambda_2$ approximately equal to $Q$, whereas, in comparison, the rest eigenvalues are small in absolute value. Furthermore, the corresponding eigenvectors $\varphi_1, \varphi_2$ will be concentrated on the $\{u,v\}$ subspace, and by orthogonality of the eigenvectors we should expect $\sgn(\varphi_1(u)\varphi_1(v)) = -\sgn(\varphi_2(u) \varphi_2(v))$.  Thus, it is natural to consider the process at time $t_0=\frac{\pi}{\lambda_1-\lambda_2}$. With this choice, we can estimate
\begin{multline}\label{ineq:p(t)}
|U(t_0)_{u,v}| \geq |\varphi_1(u)\varphi_1(v)-\varphi_2(u)\varphi_2(v)|-\sum_{j = 3}^n |\varphi_j(u)\varphi_j(v)|\\
 \geq |\varphi_1(u)\varphi_1(v)-\varphi_2(u)\varphi_2(v)|-\sqrt{(1-\varphi_1(u)^2-\varphi_2(u)^2)(1-\varphi_1(v)^2-\varphi_2(v)^2)}
\end{multline}

Showing that this last expression can be arbitrarily close to 1 by selecting an appropriate potential $Q$ is going to be our main goal. 


\section{Estimates}\label{sec:est}

In this section, we establish the main estimates and introduce the key notions to bound the values appearing in \eqref{ineq:p(t)}.

\subsection{Eigenvector extension}\label{sec:extension}

Our main tool for studying $\varphi_1$ and $\varphi_2$ for large values of $Q$ is adapted from~\cite{Tunneling_2012}.
\begin{definition} Given a graph $G(V,E)$ and $u,v \in V$, let $\cP_{xy}$ denote the set of walks in $G$ that start at $x$ and end at $y$ and avoid $u$ and $v$, except for possibly the first and last vertex. In the case of $x=y$, we do not include the 0-length walk. Note that this set is typically infinite. For a path $P \in \cP_{xy}$, let $|P|$ denote its length, i.e., its number of edges.
Finally, define a function $Z_{xy}(\lambda)$ via the formula
\[Z_{xy}(\lambda) = \sum_{P \in \cP_{xy}} \lambda^{-|P|}.\]

We will often use $n_k(xy) = |\{ P \in \cP_{xy} : |P| = k\}|$ to denote the number of allowed walks of length exactly $k$ from $x$ to $y$. Then we can rewrite $Z$ as a power series in $\lambda^{-1}$ as 

\[ Z_{xy}(\lambda) = \sum_{k=0}^{\infty} \frac{n_k(xy)}{\lambda^k} =  \sum_{k = d(x,y)}^\infty  \frac{n_k(xy)}{ \lambda^k} .\]
We omit the dependence on $u$ and $v$ from the notation in both $\cP$ and $Z$, but this will not lead to confusion as $u,v$ will be fixed throughout the paper.
\end{definition}

\begin{rem}\label{rem:z_bound} If the maximum degree of $G$ is $m$, then clearly $n_k(xy) \leq m^k$. Hence $Z_{xy}(\lambda)$ is an absolutely convergent power series for $\lambda > m$. In particular,
\[ Z_{xy}(\lambda)  \leq \sum_{k = d(x,y)}^\infty \frac{m^k}{\lambda^k} \leq \left(\frac{m}{\lambda}\right)^{d(x,y)} \frac{1}{1-\frac{m}{\lambda}}.\]
\end{rem}

Now we have the following theorem, whose proof is based on~\cite{Tunneling_2012}.
\begin{theorem}[\cite{quantifying}]\label{thm:z}
Let $G$ be a connected graph $u,v \in V(G)$, and $H$ as above. Let $m$ denote the maximum degree of $G$ and let $\lambda > m$. Let further $\mu, \nu$ be given real numbers. Then $H$ has an eigenvector $\varphi$ with eigenvalue $\lambda$ and $\varphi(u)=\mu, \varphi(v)=\nu$ if and only if
\begin{align}\label{mtx:Z}
    \begin{bmatrix}
        Z_{uu}(\lambda) & Z_{uv}(\lambda)\\
        Z_{uv}(\lambda) & Z_{vv}(\lambda)
    \end{bmatrix}
    \begin{bmatrix}
        \mu\\
        \nu
    \end{bmatrix}=\left(1-\frac{Q}{\lambda}\right)
    \begin{bmatrix}
        \mu\\
        \nu
    \end{bmatrix}
\end{align}
When $\varphi$ exists it is unique and satisfies 
\begin{equation}\label{eq:phi_exp} 
\forall x\neq u,v : \varphi(x) = \mu Z_{xu}(\lambda) + \nu Z_{xv}(\lambda).
\end{equation}
\end{theorem}
\begin{proof}
First consider a function $f: V(G)\rightarrow \mathbb{R}$ such that $f(u)=\mu$, $f(v)=\nu$, and $Hf(x)=\lambda f(x)$ for all $x\in V(G)\backslash \{u,v\}$. For any $x \neq u,v$ we then have 
\[ f(x)=\frac{1}{\lambda}Hf(x) =\sum_{y \sim x}  \frac{f(y)}{\lambda} .\]  We iterate this for all $y\neq u,v$ appearing in the above sum for up to $k$ steps. Whenever $u$ or $v$ shows up during the iteration, we keep that term as it is and only iterate on the other terms. Then we get 
\[ f(x) = \sum_{y \neq u,v} \sum_{P \in \cP_{xy} : |P| = k} \frac{f(y)}{\lambda^k} + f(u) \sum_{P \in \cP_{xu} : |P| \leq k} \frac{1}{\lambda^{|P|}} +f(v) \sum_{P \in \cP_{xv} : |P| \leq k} \frac{1}{\lambda^{|P|}}.\]
When $k \to \infty$, we see from Remark~\ref{rem:z_bound} that the first sum on the right hand side goes to 0 because it is bounded by $\left(\frac{m}{\lambda}\right)^k \sum_{y \neq u,v} |f(y)|$ since $\lambda > m$. The other two terms converge to $f(u) Z_{xu}(\lambda)$ and $f(v) Z_{xv}(\lambda)$ respectively. So we get that  
\begin{align}\label{function:f(x)}
    \forall x\neq u,v : f(x)= \mu Z_{xu}(\lambda) + \nu Z_{xv}(\lambda); f(u) = \mu; f(v) = \nu, 
\end{align} in particular we see that \eqref{eq:phi_exp} holds.

Conversely, it is easy to see that \eqref{function:f(x)} indeed defines a function $f: V(G)\rightarrow \mathbb{R}$ such that $f(u)=\mu$, $f(v)=\nu$, and $Hf(x)=\lambda f(x)$ for all $x\in V(G)\backslash \{u,v\}$. Hence there is always a unique function $f: V(G) \to \mathbb{R}$ that satisfies these conditions.

Thus, to finish the proof, it is sufficient to determine, when is $f$ actually an eigenvector of $H$. It happens exactly when $\lambda f(x) = Hf(x)$ is also satisfied at $x= u,v$,
\begin{align}\label{eq:lambdaf=hf}
    \lambda f(u)=Qf(u)+\sum_{x\sim u}f(x) \text{  ;  } \lambda f(v)=Qf(v)+\sum_{x\sim v}f(x)
\end{align}
Dividing both sides of equations \eqref{eq:lambdaf=hf} by $\lambda$, rearranging, and plugging in \eqref{function:f(x)} for each $f(x)$ in the sums, we get the following equality for both  $z=u$ and $z=v$:
\begin{align}\label{eq:f(v)}
    \left(1-\frac{Q}{\lambda}\right)f(z) = \mu \cdot \sum_{x \sim z} \frac{1}{\lambda} Z_{xu}(\lambda) + \nu \cdot \sum_{x \sim z} \frac{1}{\lambda}Z_{xv}(\lambda) = \mu Z_{zu}(\lambda) + \nu Z_{zv}(\lambda)   
\end{align} which, for $z=u,v$ gives exactly the $2\times 2$ linear system \eqref{mtx:Z}.
\end{proof}

\begin{rem}
As we see from the preceding theorem, for two real numbers $\mu$ and $\nu$ to be the entries $\varphi(u)$, $\varphi(v)$ of the eigenvector $\varphi$ of the matrix $H = A + Q \cdot D_{u,v}$ corresponding to eigenvalue $\lambda$, they should satisfy the condition of forming an eigenvector of the $2\times 2$ matrix described in \eqref{mtx:Z}. 

A particularly nice feature of \eqref{mtx:Z} is that the only dependence on $Q$ is in the factor on the right-hand side. So, for any given $\lambda > m$ one can consider the $2\times 2$ matrix on the left, compute an eigenvector $(\mu,\nu)$ and corresponding eigenvalue $\rho < 1$, and then choose $Q = \lambda(1-\rho)$. For this choice of $Q$, then $\lambda$ is guaranteed to be an eigenvalue of $H$ and the corresponding eigenvector $\varphi$, scaled appropriately, will satisfy $\varphi(u) = \mu$ and $\varphi(v) = \nu$.
\end{rem}

\begin{corollary}\label{cor:phi_bound}
Using the previous notation, if $\lambda \geq C \cdot m^{3/2}$ is an eigenvalue of $H$ for some $C > 2$ then the corresponding eigenvector satisfies
\begin{equation} \varphi(u)^2 + \varphi(v)^2 \geq \norm{\varphi}^2 \left(1 - \frac{22}{C^2}\right)
\end{equation}
\end{corollary}

\begin{proof}
 First we bound $f(x)$ pointwise using Remark~\ref{rem:z_bound}:
\[ |f(x)| \leq (|\mu| + |\nu|)\max\{ Z_{xu}(\lambda), Z_{xv}(\lambda)\} \leq \frac{|\mu|+|\nu|}{1-\frac{m}{\lambda}}\left(\frac{m}{\lambda}\right)^{r(x)} \] where $r(x) = \min\{ d(x,u), d(x,v)\}$. This implies
\[ \sum_{x \neq u,v} f(x)^2 \leq \left( \frac{|\mu|+|\nu|}{1-\frac{m}{\lambda}} \right)^2 \sum_{x\neq u,v} \left(\frac{m}{\lambda}\right)^{2r(x)} \leq 2(\mu^2+\nu^2)\left(1-\frac{m}{\lambda}\right)^{-2}\sum_{k=1}^\infty \left(\frac{m}{\lambda}\right)^{2k} |\{ x: r(x)=k\}|.\] In a graph of maximum degree $m$ there can be no more than $m^k$ nodes at distance $k$ from any given node, hence $|\{ x: r(x)=k\}| \leq 2m^k$. This yields
\[ \frac{\sum_{x \neq u,v} f(x)^2}{\mu^2+\nu^2} \leq \frac{4}{\left(1-\frac{m}{\lambda}\right)^2} \sum_{k=1}^\infty \left(\frac{m^3}{\lambda^2}\right)^k \leq \frac{4}{ (1-\frac{1}{C \sqrt{m}})^2 C^2 (1-\frac{1}{C^2})}\leq \frac{64}{3 C^2} \] and hence 
\[ \norm{\varphi}^2 \leq (\mu^2+ \nu^2)\left(1+ \frac{64}{3C^2}\right) \leq (\varphi(u)^2 + \varphi(v)^2)  \left(1 + \frac{22}{C^2}\right)\]  whence the claim follows.
\end{proof}

\subsection{Cospectrality}\label{se:cospectral}

Let us recall the following definitions and theorem from~\cite{Tunneling_2012}.
\begin{definition}\label{def:cos}
    Two vertices, $u,v\in V(G)$ are \textbf{$k$-cospectral} if, for any positive integer $t\leq k$ the number of closed walks of length $t$ at $u$ and $v$ are equal to each other, or equivalently, if $\forall t \leq k : n_t(uu) = n_t(vv)$. The \textbf{cospectrality} of $u$ and $v$, denoted by $co(u,v)$ is the maximum number $m$ such that $u$ and $v$ are $m$-cospectral.
\end{definition}
\begin{definition}\label{def:TC}
    Define the tunneling coefficient from state $u$ to $v$ to be
    \begin{align}\label{eq:infsup}
    TC(u,v)=\liminf_{Q\rightarrow \infty}{\sup_{t\geq 0} p(t)}.
\end{align}
    We say there is 
    \begin{enumerate}
        \item asymptotic tunneling (strong state transfer) if $TC(u,v)=1$
        \item partial tunneling if $0<TC(u,v)<1$
        \item no tunneling if $TC(u,v)=0$
    \end{enumerate}
\end{definition}

The relationship between cospectrality and tunneling behavior was studied in~\cite{Tunneling_2012} where the following was shown for the random walk transition matrix in place of the adjacency matrix (however the proofs carry over without much modification to the current case):

\begin{theorem}\cite{Tunneling_2012}\label{thm:tunneling}
    Let there be two vertices $u$ and $v$ and let $d=\text{dist}(u,v)$, then the following is true:
    \begin{enumerate}
        \item if $co(u,v)\geq d$ then there is asymptotic tunneling
        \item if $co(u,v)= d-1$ then there is partial tunneling
        \item if $co(u,v)< d-1$ then there is no tunneling
    \end{enumerate}
\end{theorem}

This theorem suggests that $co(u,v) \geq d$ is a reasonable assumption in order to achieve strong state transfer using large values of $Q$. We have seen that, for large enough $Q$, strong state transfer depends on $|\varphi_j(u)| \approx |\varphi_j(v)|$ for $j = 1,2$. We will now show that $co(u,v) \geq d$ indeed implies this in a quantitative way.

\begin{lemma}\label{lemma:ratio}
Using the notation from Section~\ref{sec:extension}, if $\lambda >  m$ is an eigenvalue of $H$ and $c \geq d$ then the corresponding eigenvector $\varphi$ satisfies
\[ \left| \frac{\varphi(u)}{\varphi(v)} - \frac{\varphi(v)}{\varphi(u)}\right| \leq \frac{m^{c+1}}{\lambda^{c-d}(\lambda - m)} \] where $c = co(u,v)$ and $d=\text{dist}(u,v)$.
\end{lemma}

\begin{proof}
Let $\mu = \varphi(u)$ and $\nu = \varphi(v)$. By Theorem \ref{thm:z}, 
$\begin{bmatrix}
    \mu\\
    \nu
\end{bmatrix}$ is an eigenvector of the $Z$ matrix in \eqref{mtx:Z}, and dividing the two resulting equations to cancel the common eigenvalue yields
\[ \frac{\mu}{\nu} = \frac{\mu Z_{uu}(\lambda) + \nu Z_{uv}(\lambda)}{\mu Z_{uv}(\lambda) + \nu Z_{vv}(\lambda)}
.\] After rearranging we get 
\[ (\mu^2 - \nu^2)Z_{uv}(\lambda) = \mu \nu (Z_{uu}(\lambda) - Z_{vv}(\lambda)),\] and finally
\begin{align*}
   \left| \frac{\mu}{\nu}-\frac{\nu}{\mu}\right|=
   \left| \frac{Z_{uu}(\lambda)-Z_{vv}(\lambda)}{Z_{uv}(\lambda)}\right|
    =\frac{\left|\sum_{k=2}^\infty \frac{n_k(u,u)}{\lambda^{k}}-\sum_{k=2}^\infty\frac{n_k(v,v)}{\lambda^{k}}\right|}{\sum_{k=d}^\infty\frac{n_k(u,v)}{\lambda^{k}}}
\end{align*}
Here the denominator is at least $1/\lambda^d$ while the $k \leq c$ terms cancel in the numerator, hence 
\[ 
\left| \frac{\mu}{\nu}-\frac{\nu}{\mu}\right| \leq \lambda^d \sum_{k>c} \frac{|n_k(uu) - n_k(vv)|}{\lambda^k} \leq \lambda^d \sum_{k>c} \left(\frac{m}{\lambda}\right)^k = \frac{m^{c+1}}{\lambda^{c-d}(\lambda - m)} 
\] as claimed.
\end{proof}


\section{Bounding Fidelity in the Presence of Quantum Tunneling}\label{se:mainthm}

In this section, we prove our main estimate: a lower bound of the transfer strength of the quantum walk between two vertices that exhibit asymptotic state transfer as defined in Definition~\ref{def:TC}. This bound is a function of the potential $Q$ that is applied to the nodes $u$ and $v$, along with the maximum degree of the graph, and in certain cases, the distance and cospectrality of the nodes $u,v$. 
\begin{theorem}\label{main}
    Consider a simple connected graph $G=(V, E)$ with maximum degree $m$. Let $u$ and $v$ be two vertices exhibiting asymptotic state transfer, meaning their cospectrality $c = co(u,v)$ is at least their distance $d = d(u,v)$. After adding a potential $Q$ to both vertices, the transfer strength from $u$ to $v$ at time $t_0 = \pi/(\lambda_1 - \lambda_2)$ satisfies
 \[       p(t_0)\geq 1 - 8 \max\left\{ \frac{22m^3}{(Q-m)^2},\frac{m^{c+1}}{(Q-2m)^{c-d+1}}  \right\} \]
\end{theorem}

Theorem~\ref{thm:main} follows directly from Theorem~\ref{main} by a simple calculation.

\begin{proof}
Since $Q > 2m$, Gershgorin's Circle Theorem implies that there are exactly two eigenvalues in the interval $[Q-m, Q+m]$ and all other eigenvalues are in the $[-m,m]$ interval. Let $\lambda_1, \lambda_2$ denote the two eigenvalues in $[Q-m, Q+m]$ and let $\varphi_1, \varphi_2$ denote the corresponding unit length eigenvectors.

For simplicity, we denote the $u,v$ entries in the normalized eigenvectors in the following manner from now on:
\[\varphi_1(u,v)=\begin{bmatrix}
    \mu_1\\
    \nu_1
    \end{bmatrix}, 
    \varphi_2(u,v)=\begin{bmatrix}
    \mu_2\\
    \nu_2
    \end{bmatrix}.\]

Without loss of generality we may assume $\mu_1\geq\nu_1$. From Corollary~\ref{cor:phi_bound} and Lemma~\ref{lemma:ratio}, we know that 
\begin{align*}
        1 \geq \mu_1^2+\nu_1^2&\geq 1- \frac{22m^3}{\lambda_1^2} \geq 1-\frac{22m^3}{(Q-m)^2}\\
        1\leq\frac{|\mu_1|}{|\nu_1|}&\leq 1+ \frac{m^{c+1}}{(\lambda_1 - m)^{c-d+1}} \leq 1+\frac{m^{c+1}}{(Q-2m)^{c-d+1}}.
\end{align*}
 Thus, with $\ep' = \max\{\frac{22m^3}{(Q-m)^2},\frac{m^{c+1}}{(Q-2m)^{c-d+1}}\}$ we get
 \[
 1- \ep' \leq \mu_1^2+\nu_1^2 \leq \nu_1^2(1+(1+\ep')^2)
 \]
 so 
 \[ \nu_1^2 \geq \frac{1-\ep'}{2+2\ep'+\ep'^2} \geq \frac{1}{2} - \ep'.\] 
Hence
\[ \frac{1}{2}-\ep' \leq \nu_1^2 \leq \mu_1^2 \leq \frac{1}{2}+\ep' \]
The same is clearly true for $\mu_2$ and $\nu_2$. Due to Perron-Frobenius Theorem, $\mu_1$ and $\nu_1$ have the same sign. However,  $\mu_2,\nu_2$ must have opposite signs. To see this, note that $\varphi_1$ and $\varphi_2$ are orthogonal to each other, hence $0 = \mu_1 \mu_2 + \nu_1 \nu_2 + \sum_{x \neq u,v} \varphi_1(x)\varphi_2(x)$. 
This implies 
\begin{multline*}
 |\mu_1 \mu_2 + \nu_1 \nu_2| = \left| \sum_{x \neq u,v} \varphi_1(x)\varphi_2(x)\right| \leq \sqrt{\sum_{x \neq u,v} \varphi_1(x)^2}\sqrt{\sum_{x\neq u,v} \varphi_2(x)^2}\\ = \sqrt{1-\mu_1^2-\nu_1^2}\sqrt{1-\mu_2^2 - \nu_2^2} \leq 2\ep' 
\end{multline*}
But if $\mu_2$ and $\nu_2$ had the same sign then the left-hand side would be at least $1-2\ep'$ which implies $\ep' \geq 1/4$ and there is nothing to prove. Therefore, employing \eqref{ineq:p(t)} with $t_0 = \frac{\pi}{\lambda_1 - \lambda_2}$, the transfer strength is at least
\begin{align*} 
p(t_0) \geq (1-2\ep' - \sqrt{(2\ep')(2\ep')})^2 \geq 1-8\ep'
\end{align*}
\end{proof}


\section{Readout Time}\label{sec:readout}

We have seen in Theorem~\ref{main} that at $t_0 = \frac{\pi}{\lambda_1-\lambda_2}$ the transfer strength $p(t_0)$ is guaranteed to be at least $1-\ep$ as long as $Q$ satisfies the lower bound set forth in Theorem~\ref{thm:main}. Unfortunately, the scaling of $t_0$ is always exponential in the distance $d = d(u,v)$ between the nodes $u$ and $v$. Nonetheless, our main goal in this section is to give an upper bound on $t_0$, or equivalently, a lower bound on $\lambda_1 - \lambda_2$, the gap between the two largest eigenvalues of the Hamiltonian. This will be done by adapting our methods from~\cite{quantifying}. The first part of Theorem~\ref{thm:readout} follows from the subsequent statement.

\begin{theorem}\label{thm:t0bound}
Under the assumptions of Theorem~\ref{thm:main} the following bound holds for $t_0 = \frac{\pi}{\lambda_1-\lambda_2}$:
\[ t_0 \leq 2\pi (Q+m)^{d-1} \]
where $d = d(u,v)$, the distance of the source and target nodes in the graph.
\end{theorem}

\begin{proof}
Our starting point is the $Z$ matrix \eqref{mtx:Z}, and we apply it to both $\varphi_1$ and $\varphi_2$. Using an indexed version of the notation there we find
\begin{align*}
(\lambda_1 - Q)\mu_1 =& \lambda_1 Z_{uu}(\lambda_1)\mu_1 + \lambda_1 Z_{uv}(\lambda_1)\nu_1 \\
(\lambda_2 - Q)\mu_2 =& \lambda_2 Z_{uu}(\lambda_2)\mu_2 + \lambda_2 Z_{uv}(\lambda_2)\nu_2 
\end{align*}
\end{proof}
which gives, after further rearrangement and taking difference of the two equations:
\begin{equation}\label{eq:la1la2} (\lambda_1 - \lambda_2) - (\lambda_1 Z_{uu}(\lambda_1) - \lambda_2 Z_{uu}(\lambda_2)) = \lambda_1 \frac{\nu_1}{\mu_1} Z_{uv}(\lambda_1) - \lambda_2 \frac{\nu_2}{\mu_2} Z_{uv}(\lambda_2). \end{equation}
Note that, according to the proof of Theorem~\ref{thm:main},  $\nu_1/\mu_1 > 0 > \nu_2/\mu_2$ and hence all terms on the right-hand side are positive. Hence, we can lower bound it by the first respective terms in the series of $Z_{uv}$. Also, expanding the left-hand side using the definition of $Z_{uu}$ we get the following inequality
\[ (\lambda_1-\lambda_2)  + \sum_{k=1}^\infty n_{k+1}(uu) \left(\frac{1}{\lambda_2^k} - \frac{1}{\lambda_1^k}\right) \geq \left( \left| \frac{\nu_1}{\mu_1}\right|\frac{1}{\lambda_1^{d-1}} +\left| \frac{\nu_2}{\mu_2}\right|\frac{1}{\lambda_2^{d-1}} \right) \geq \frac{1}{\lambda_1^{d-1}}.\] 
Here $d$ denotes $d(u,v)$, and in the last step, we used that  both $|\nu_1/\mu_1| \geq 1-\ep$ and $|\nu_2/\mu_2| \geq 1-\ep$, so their sum is clearly at least 1. We can further rewrite the left-hand side to get
\[ (\lambda_1 - \lambda_2) \left( 1+ \sum_{k=1}^\infty n_{k+1}(uu) \sum_{j=1}^{k}\frac{1}{\lambda_1^j \lambda_2^{k+1-j}}\right) \geq \frac{1}{\lambda_1^{d-1}}.\]
Now using the usual $n_k(uu)\leq m^k$ bound we find
\begin{multline*}  \frac{1}{\lambda_1^{d-1}} \leq (\lambda_1 - \lambda_2) \left(1+ \sum_{k=1}^\infty \frac{k\cdot m^{k+1}}{\lambda_2^{k+1}}\right)=(\lambda_1 - \lambda_2) \left(1+ \frac{1}{1-m/\lambda_2} \sum_{k=2}^\infty \frac{m^{k}}{\lambda_2^{k}}\right)=\\ (\lambda_1 - \lambda_2) \left(1+ \frac{m^2}{(\lambda_2 - m)^2}\right) \leq (\lambda_1- \lambda_2)\left(1+ \frac{m^2}{(Q-2m)^2}\right) \leq  (\lambda_1- \lambda_2)\left(1+ \frac{1}{m^2}\right),
\end{multline*}
hence, finally,
\[ t_0 = \frac{\pi}{\lambda_1-\lambda_2} \leq \lambda_1^{d-1}\frac{m^2+1}{m^2} \leq 2 (Q+m)^{d-1} \]
\subsection{Readout interval}

In this section, we establish the second part of Theorem \ref{thm:readout} regarding the length of the interval centered around $t_0$. 

\begin{theorem}\label{thm:readout_interval}
Let $\ep >0$ be fixed and assume $Q$ satisfies the bounds set forth in Theorem~\ref{thm:main}. Then for all $t \in [\frac{\pi - \delta}{\lambda_1-\lambda_2}, \frac{\pi + \delta}{\lambda_1-\lambda_2} ] $, the following bound is true:
\[ p(t) \geq 1-\ep-2\delta \]
\end{theorem}

\begin{proof}
A simple computation shows that the lower bound on $Q$ implies that $\ep' = \max\{\frac{22m^3}{(Q-m)^2},\frac{m^{c+1}}{(Q-2m)^{c-d+1}}\} \leq \ep/8$. (This is the same computation that shows Theorem~\ref{main} implies Theorem~\ref{thm:main}.) 
A straightforward generalization of \eqref{ineq:p(t)} shows that
\[ |U(t)_{u,v}| \geq |e^{it(\lambda_1-\lambda_2)} \varphi_1(u)\varphi_1(v)-\varphi_2(u)\varphi_2(v)|-\sqrt{(1-\varphi_1(u)^2-\varphi_2(u)^2)(1-\varphi_1(v)^2-\varphi_2(v)^2)}.\]
Notice that for real numbers $a > 0 > b$ and angle $ \pi - \delta < \rho < \pi+\delta$ we have 
\begin{multline*} |a + e^{i\rho}b|^2 = a^2 + b^2 +ab(e^{i\rho}+e^{-i\rho}) = (a-b)^2 + ab(2+e^{i\rho}+e^{-i\rho}) \\ \geq (a-b)^2 - 2ab(1-\cos \delta)\geq (a-b)^2 -\delta^2 ab\end{multline*} and thus
\[ |a + e^{i\rho}b| \geq |a-b| - \delta \sqrt{ab}.\]
Combining this with estimates from the proof of Theorem~\ref{main}, and noticing that $\pi - \delta \leq t(\lambda_1-\lambda_2) \leq \pi+\delta$, we get 
\begin{multline*} |U(t)_{u,v}| \geq  | \varphi_1(u)\varphi_1(v)+ e^{it(\lambda_1-\lambda_2)}\varphi_2(u)\varphi_2(v)| - 2\ep'  \\ 
\geq |\varphi_1(u)\varphi_1(v) - \varphi_2(u)\varphi_2(v)| - \delta \sqrt{|\varphi_1(u)\varphi_1(v)\varphi_2(u)\varphi_2(v)|} -2\ep' 
\geq 1-4\ep' - \delta
\end{multline*}
and thus
\[ p(t) = \left|U(t)_{u,v}\right|^2 \geq (1-4\ep' -\delta)^2 \geq 1 - 8\ep'- 2\delta = 1-\ep - 2\delta. \]
\end{proof}


\bibliographystyle{plain}
\bibliography{ref}

\begin{thebibliography}{10}

\bibitem{Bose_2003}
Sougato Bose.
\newblock Quantum communication through an unmodulated spin chain.
\newblock {\em Physical Review Letters}, 91(20), November 2003.

\bibitem{casaccino2009quantum}
Andrea Casaccino, Seth Lloyd, Stefano Mancini, and Simone Severini.
\newblock Quantum state transfer through a qubit network with energy shifts and fluctuations, 2009.

\bibitem{Chen_2016}
Xining Chen, Robert Mereau, and David~L. Feder.
\newblock Asymptotically perfect efficient quantum state transfer across uniform chains with two impurities.
\newblock {\em Physical Review A}, 93(1), January 2016.

\bibitem{Christandl_2005}
Matthias Christandl, Nilanjana Datta, Tony~C. Dorlas, Artur Ekert, Alastair Kay, and Andrew~J. Landahl.
\newblock Perfect transfer of arbitrary states in quantum spin networks.
\newblock {\em Physical Review A}, 71(3), March 2005.

\bibitem{kay2010}
A.~Kay.
\newblock Perfect, efficient, state transfer and its applications as a constructive tool.
\newblock {\em Int. J. Quantum Inform}, 8(4):641, 2010.

\bibitem{path}
Stephen Kirkland and Christopher~M. van Bommel.
\newblock {State transfer on paths with weighted loops}.
\newblock {\em Quant. Inf. Proc.}, 21(6):209, 2022.

\bibitem{Tunneling_2012}
Yong Lin, G{\'{a}}bor Lippner, and Shing-Tung Yau.
\newblock Quantum tunneling on graphs.
\newblock {\em Communications in Mathematical Physics}, 311(1):113--132, 02 2012.

\bibitem{xx_magnetic_2012}
Thorben Linneweber, Joachim Stolze, and G\"{o}tz~S. Uhrig.
\newblock Perfect state transfer in xx chains induced by boundary magnetic fields.
\newblock {\em International Journal of Quantum Information}, 10(03):1250029, 2012.

\bibitem{quantifying}
Gabor Lippner and Yujia Shi.
\newblock Quantifying state transfer strength on graphs with involution, 2023.

\bibitem{Shi_2005}
Tao Shi, Ying Li, Zhi Song, and Chang-Pu Sun.
\newblock Quantum-state transfer via the ferromagnetic chain in a spatially modulated field.
\newblock {\em Physical Review A}, 71(3), March 2005.

\end{thebibliography}

\end{document}